\documentclass[final,5p,times,twocolumn]{elsarticle}

\usepackage{amssymb}
\usepackage{amsmath}
\usepackage{amsthm}
\usepackage{graphicx}
\usepackage{enumerate}
\usepackage{bm}
\newtheorem{proposition}{Proposition}
\newtheorem{theorem}{Theorem}

\newtheorem{remark}{Remark}
\newtheorem{definition}{Definition}

\begin{document}

\begin{frontmatter}



\title{Port-Hamiltonian Systems with Dissipation Potential: Modelling and Trajectory Tracking Control} 


\author[add1]{Jinjun Jia}
\author[add1]{Yuchen Liao}
\author[add1]{Kang An}
\author[add1,add2,add3]{Xun Yan}
\author[add1,add2,add3]{Tiedong Zhang}
\author[add1,add2,add3]{Dapeng Jiang\corref{cor1}}
\ead{jiangdp5@mail.sysu.edu.cn}
\cortext[cor1]{Corresponding author}
\address[add1]{School of Ocean Engineering and Technology, Sun Yat-sen University $\&$ Southern Marine Science and Engineering Guangdong Laboratory (Zhuhai). Zhuhai, 519000, China}
\address[add2]{Zhuhai Research Center, Hanjiang National Laboratory. Zhuhai, 519000, China}
\address[add3]{Guangdong Provincial Key Laboratory of Information Technology for Deep Water Acoustics. Zhuhai, 519082, China}

\begin{abstract}
Port-Hamiltonian systems (PHS) and interconnection and damping assignment passivity-based control (IDA-PBC) have achieved broad success in modelling and stabilisation of physical systems. However, the absence of a dedicated scalar potential for the momentum channel forces any modification of the momentum-dependent dynamics to proceed indirectly through the interconnection and damping matrices, rendering the matching partial differential equation (PDE) difficult to solve and complicating extensions to trajectory tracking. This paper proposes a port-Hamiltonian system with dissipation potential (PHS-DP), in which the damping matrix is replaced by scalar convex dissipation potentials, providing independent scalar objects for the momentum and auxiliary state channels and restoring the variational symmetry between stored and dissipated energy. Building on this framework, Dual Potential Shaping Control (DPSC) achieves trajectory tracking by sequentially shaping the potential energy and dissipation potentials without modifying the interconnection structure. Contraction of the closed-loop cascade is established via a hierarchical contraction argument, and the matching condition is satisfied automatically for any admissible choice of shaped potentials, requiring no PDE to be solved. In contrast to existing PDE-free energy shaping approaches, which achieve this by abandoning the port-Hamiltonian closed-loop structure and sacrificing physical interpretability, the proposed framework preserves the interconnection structure and retains a transparent energy-based interpretation at every stage of the design. Validation on a magnetic levitation system demonstrates tracking performance comparable to timed IDA-PBC with substantially reduced design complexity.
\end{abstract}


\end{frontmatter}



\section{Introduction}
\label{sec:intro}
Port-Hamiltonian systems (PHS) extend classical Hamiltonian mechanics to open, dissipative systems, providing an energy-based modelling framework in which the storage, dissipation, and exchange of energy are made explicit~\cite{VanDerSchaft2014,Duindam2009}. Building on this foundation, the interconnection and damping assignment passivity-based control (IDA-PBC) method~\cite{Ortega2002} exploits the physical structure of the system to achieve stabilisation without nonlinearity cancellation, yielding controllers that are physically interpretable and inherently robust to model uncertainty, and has been successfully applied to mechanical~\cite{Franco2024}, electromechanical~\cite{Rodriguez2003}, and process systems~\cite{Doerfler2009}.

However, the application of IDA-PBC faces significant obstacles. The method addresses port-Hamiltonian systems with a Poisson structure, whose states decompose into generalised coordinates $\mathbf{q}$, generalised momenta $\mathbf{p}$, and additional states $\mathbf{s}$~\cite{VanderSchaft2000}. Two structurally distinct formulations arise depending on whether $\mathbf{s}$ is present. When $\mathbf{s}$ is absent, as in standard mechanical PHS whose Hamiltonian separates into kinetic and potential energy, the parameterised IDA-PBC requires solving a kinetic energy PDE and a potential energy PDE to obtain the desired kinetic and potential energies; for underactuated systems, both must be shaped simultaneously, and the kinetic energy PDE is a nonlinear PDE in the desired inertia matrix $M_d$ that is generally difficult to solve. When $\mathbf{s}$ is present, as in electromechanical systems such as magnetic levitation whose interconnection matrix contains zero rows corresponding to Casimir directions, the non-parameterised formulation applies: $\mathbf{J}_d$ and $\mathbf{R}_d$ must first be modified to eliminate these Casimir directions, after which the matching PDE for $H_d$ can be solved. In both formulations, IDA-PBC is subject to two fundamental limitations: (i) the matching PDE is in general difficult to solve analytically; and (ii) the energy interpretation of the closed-loop Hamiltonian $H_d$ becomes unclear. Constructive procedures that avoid solving the matching PDE have been proposed~\cite{Borja2016} by abandoning the objective of preserving the port-Hamiltonian structure in closed loop; however, this sacrifices the physical interpretability that is central to the PHS framework.

Extending IDA-PBC to trajectory tracking introduces additional difficulties, since the error dynamics of nonlinear PHS generically does not admit a port-Hamiltonian structure~\cite{Yaghmaei2017}, rendering the standard approach of recasting tracking as a stabilisation problem incompatible with the IDA-PBC framework. Fujimoto et al.~\cite{Fujimoto2003} circumvented this obstacle via a generalised canonical transformation, though at the cost of solving an additional PDE and without constituting a genuine extension of IDA-PBC. Yaghmaei and Yazdanpanah~\cite{Yaghmaei2015,Yaghmaei2017} proposed the timed IDA-PBC (tIDA-PBC), which extends IDA-PBC to trajectory tracking through contraction theory, circumventing error dynamics entirely. This framework was subsequently generalised to time-varying~\cite{Barabanov2019} and state-modulated~\cite{Yaghmaei2023} matrix structures. Nevertheless, tIDA-PBC fully inherits both fundamental limitations of IDA-PBC identified above.

The root cause of these limitations lies in the absence of a dedicated scalar potential for the cotangent space in the standard PHS formulation: kinetic energy cannot serve this role. Consequently, no scalar object exists through which the momentum-dependent part can be shaped directly, which not only precludes momentum-space control objectives such as velocity regulation, but also forces any modification of the momentum-dependent part of $H_d$ to proceed indirectly through $\mathbf{J}_d$ and $\mathbf{R}_d$, coupling the dissipation design to the interconnection structure and rendering both the matching equation and the physical interpretation intractable. This structural deficiency is closely related to the dissipation obstacle identified in~\cite{Venkatraman2010}: coordinates affected by natural damping cannot be shaped via control by interconnection precisely because the damping matrix, rather than a scalar potential, governs the dissipative dynamics. Whereas~\cite{Venkatraman2010} addresses this obstacle by introducing alternate passive input-output pairs within the standard PHS framework, the present work takes a different approach by replacing the damping matrix itself with a scalar dissipation potential, treating dissipation as an independent physical object on equal footing with the stored energy. The classical Rayleigh dissipation function offers a more natural perspective, representing dissipation as an independent convex scalar potential whose gradient generates the dissipative force, in direct analogy to the role of the Hamiltonian in encoding stored energy. Motivated by this observation, we propose a modelling framework termed port-Hamiltonian systems with a dissipation potential (PHS-DP), in which the damping matrix is replaced by a convex scalar dissipation potential $\mathcal{F}(x)$. By treating stored and dissipated energies as two independent scalar objects with explicit physical meaning, PHS-DP restores the energy-based symmetry of the port-Hamiltonian framework and provides the conditions for constructing a scalar potential over the momentum space.

Building on PHS-DP, we propose \emph{Dual Potential Shaping Control} (DPSC) for trajectory tracking. DPSC shapes two physically transparent scalar functions: a desired potential energy that modifies the potential energy component of the system Hamiltonian to encode the desired position trajectory, and a desired dissipation potential that is modified to encode the desired momentum trajectory. The interconnection structure is preserved entirely throughout the design. Contraction of the closed-loop PHS-DP is established via a hierarchical contraction argument. Crucially, the matching condition is satisfied automatically for any admissible pair of shaped potentials, requiring no PDE to be solved, and the entire design proceeds from a physically interpretable, energy-based perspective that is transparent at every stage.

The remainder of this paper is organised as follows. Section~\ref{sec:pre} recalls the necessary preliminaries. Section~\ref{sec:model} presents the PHS-DP framework. Section~\ref{sec:control} proposes DPSC and establishes its convergence properties. Section~\ref{sec:sim} validates the approach on the magnetic levitation system. Section~\ref{sec:con} concludes the paper.

\section{Preliminaries}
\label{sec:pre}

\subsection{Port-Hamiltonian Systems}

A PHS is defined by the tuple $(H, J, \mathbf{R}, \mathbf{g})$ and governed by~\cite{VanDerSchaft2014,Duindam2009}
\begin{alignat}{2}
\dot{\mathbf{x}} &= \bigl(\mathbf{J}(\mathbf{x}) - \mathbf{R}(\mathbf{x})\bigr)\nabla H(\mathbf{x})
+ \mathbf{g}(\mathbf{x})\,\mathbf{u},
\label{eq:phs}\\[4pt]
\mathbf{y} &= \mathbf{g}^\top(\mathbf{x})\nabla H(\mathbf{x}),
\label{eq:phs_output}
\end{alignat}
where $\mathbf{x} \in \mathbb{R}^n$ is the state, $H:\mathbb{R}^n \to \mathbb{R}$ is the Hamiltonian representing the total stored energy, $\mathbf{J}(\mathbf{x}) = -\mathbf{J}^\top(\mathbf{x})$ is the interconnection matrix, $\mathbf{R}(\mathbf{x}) = \mathbf{R}^\top(\mathbf{x}) \succeq \mathbf{0}$ is the damping matrix, $\mathbf{g}(\mathbf{x}) \in \mathbb{R}^{n \times m}$ is the input map, and $(\mathbf{u}, \mathbf{y}) \in \mathbb{R}^m \times \mathbb{R}^m$ are conjugate port variables. The power balance
\begin{equation}
\dot{H} = -\nabla H^\top \mathbf{R}\,\nabla H + \mathbf{y}^\top \mathbf{u} \leq \mathbf{y}^\top \mathbf{u}
\label{eq:phs_power}
\end{equation}
confirms that~\eqref{eq:phs}--\eqref{eq:phs_output} is passive with storage function $H$ and supply rate $\mathbf{y}^\top \mathbf{u}$.

\subsection{Trajectory Tracking IDA-PBC}

The tIDA-PBC framework~\cite{Yaghmaei2015,Yaghmaei2017} is summarised as follows. A trajectory $\mathbf{x}^\star(t)$ is \emph{feasible} for~\eqref{eq:phs} if there exists $\mathbf{u}^\star(t)$ such that
\begin{equation}
\dot{\mathbf{x}}^\star =
\bigl(\mathbf{J}(\mathbf{x}^\star) - \mathbf{R}(\mathbf{x}^\star)\bigr)
\nabla H(\mathbf{x}^\star) + \mathbf{g}(\mathbf{x}^\star)\,\mathbf{u}^\star(t)
\label{eq:feasible_tida}
\end{equation}
holds for all $t \geq 0$. Suppose there exist a skew-symmetric matrix $\mathbf{J}_d$, a positive semidefinite matrix $\mathbf{R}_d$, and a function $H_d(\mathbf{x}, \mathbf{x}^\star(t))$ satisfying the contraction conditions of~\cite{Yaghmaei2017}, the matching equation
\begin{equation}
\mathbf{g}^\perp\bigl[(\mathbf{J} - \mathbf{R})\nabla H\bigr]
= \mathbf{g}^\perp\bigl[(\mathbf{J}_d - \mathbf{R}_d)\nabla H_d\bigr],
\label{eq:matching_tida}
\end{equation}
where $\mathbf{g}^\perp(\mathbf{x})$ denotes a full-rank left annihilator of $\mathbf{g}(\mathbf{x})$, and the trajectory evolution condition
\begin{equation}
\dot{\mathbf{x}}^\star = (\mathbf{J}_d - \mathbf{R}_d)\nabla H_d(\mathbf{x}^\star, \mathbf{x}^\star(t)).
\label{eq:target_tida}
\end{equation}
Then the control law
\begin{equation}
\mathbf{u} =
\bigl(\mathbf{g}^\top \mathbf{g}\bigr)^{-1}\mathbf{g}^\top
\bigl[(\mathbf{J}_d - \mathbf{R}_d)\nabla H_d
- (\mathbf{J} - \mathbf{R})\nabla H\bigr]
\label{eq:control_tida}
\end{equation}
renders the closed-loop system locally exponentially tracking of $\mathbf{x}^\star(t)$. This framework was subsequently extended to state-dependent closed-loop interconnection and damping matrices in~\cite{Yaghmaei2023}.

\subsection{Contraction Analysis}

Consider a nonlinear system $\dot{\mathbf{x}} = f(\mathbf{x}, t)$ with virtual displacement $\delta\mathbf{x}$ evolving as $\delta\dot{\mathbf{x}} = (\partial f / \partial \mathbf{x})\,\delta\mathbf{x}$~\cite{Lohmiller1998}.

\begin{definition}[\cite{Lohmiller1998}]
A region $\mathcal{C} \subseteq \mathbb{R}^n$ is a \emph{contraction region} with respect to a uniformly positive definite metric $\mathbf{M}(\mathbf{x},t) = \mathbf{\Theta}^\top(\mathbf{x},t)\mathbf{\Theta}(\mathbf{x},t) \succ \mathbf{0}$ if there exists $\beta > 0$ such that
\begin{equation}
\dot{\mathbf{M}} + \left(\frac{\partial f}{\partial \mathbf{x}}\right)^\top \mathbf{M}
+ \mathbf{M}\,\frac{\partial f}{\partial \mathbf{x}}
\preceq -2\beta \mathbf{M},
\qquad \forall\,\mathbf{x} \in \mathcal{C},\,t \geq 0.
\label{eq:contraction}
\end{equation}
\end{definition}

\begin{proposition}[\cite{Lohmiller1998}]
\label{prop:contraction}
If $\mathcal{C}$ is a contraction region and is forward invariant, then any two trajectories initialised in $\mathcal{C}$ satisfy $\|\delta\mathbf{x}(t)\|_{\mathbf{M}} \leq \|\delta\mathbf{x}(0)\|_{\mathbf{M}}\,e^{-\beta t}$. If $\mathcal{C} = \mathbb{R}^n$, contraction is global. In particular, if a feasible trajectory $\mathbf{x}^\star(t)$ lies in $\mathcal{C}$, all trajectories initialised in $\mathcal{C}$ converge exponentially to $\mathbf{x}^\star(t)$.
\end{proposition}

\subsection{Hierarchical Contraction}

\begin{proposition}[\cite{Lohmiller1998,Slotine2003}]
\label{prop:hierarchical}
Consider the cascade system
\begin{equation}
\frac{d}{dt}
\begin{pmatrix} \mathbf{x}_1 \\ \mathbf{x}_2 \end{pmatrix}
=
\begin{pmatrix} f_1(\mathbf{x}_1, t) \\ f_2(\mathbf{x}_1, \mathbf{x}_2, t) \end{pmatrix},
\label{eq:cascade_general}
\end{equation}
whose virtual dynamics takes the lower-triangular form
\begin{equation}
\frac{d}{dt}
\begin{pmatrix} \delta\mathbf{x}_1 \\ \delta\mathbf{x}_2 \end{pmatrix}
=
\begin{pmatrix}
\mathbf{F}_{11} & \mathbf{0} \\
\mathbf{F}_{21} & \mathbf{F}_{22}
\end{pmatrix}
\begin{pmatrix} \delta\mathbf{x}_1 \\ \delta\mathbf{x}_2 \end{pmatrix},
\qquad
\mathbf{F}_{ij} = \frac{\partial f_i}{\partial \mathbf{x}_j}.
\label{eq:virtual_triangular}
\end{equation}
If $\dot{\mathbf{x}}_1 = f_1(\mathbf{x}_1, t)$ is contracting with rate $\beta_1 > 0$ in metric $\mathbf{M}_1$, $\dot{\mathbf{x}}_2 = f_2(\mathbf{x}_1, \mathbf{x}_2, t)$ is contracting in $\mathbf{x}_2$ with rate $\beta_2 > 0$ in metric $\mathbf{M}_2$ uniformly in $\mathbf{x}_1$, and $\mathbf{F}_{21}$ is bounded, then~\eqref{eq:cascade_general} is contracting under the composite metric $\mathrm{blkdiag}(\mathbf{M}_1,\, \varepsilon \mathbf{M}_2)$ for sufficiently small $\varepsilon > 0$. The result extends by recursion to hierarchies of arbitrary depth~\cite{Slotine2003}.
\end{proposition}

\section{Port-Hamiltonian Systems with Dissipation Potential}
\label{sec:model}

In the standard PHS formulation, dissipation is encoded via the damping matrix $\mathbf{R}(\mathbf{x}) = \mathbf{R}^\top(\mathbf{x}) \succeq \mathbf{0}$, which enters as a primitive structural object rather than being derived from an underlying scalar potential. This breaks the variational symmetry between stored and dissipated energy and, as argued in Section~\ref{sec:intro}, is the root cause of the fundamental limitations of IDA-PBC. We propose an alternative formulation in which dissipation is generated by a scalar dissipation potential, thereby restoring this symmetry.

We consider systems whose state decomposes as $\mathbf{x} = (\mathbf{q}^\top, \mathbf{p}^\top, \mathbf{s}^\top)^\top \in \mathbb{R}^{n_q} \times \mathbb{R}^{n_p} \times \mathbb{R}^{n_s}$, where $\mathbf{q} \in \mathbb{R}^{n_q}$ denotes the generalised coordinates, $\mathbf{p} \in \mathbb{R}^{n_p}$ the generalised momenta, and $\mathbf{s} \in \mathbb{R}^{n_s}$ the additional states. This decomposition encompasses a broad class of physical systems; the purely mechanical case corresponds to $n_s = 0$.

\begin{definition}
\label{defn:phs_dp}
A \emph{port-Hamiltonian system with dissipation potential} (PHS-DP) is defined by the tuple $(H_{qp}, H_s, \mathbf{J}, \mathcal{F}_p, \mathcal{F}_s, \mathbf{g})$ and governed by
\begin{equation}
\begin{pmatrix}
\dot{\mathbf{q}} \\[4pt] \dot{\mathbf{p}} \\[4pt] \dot{\mathbf{s}}
\end{pmatrix}
=
\underbrace{
\begin{pmatrix}
\mathbf{J}_{qq} & \mathbf{J}_{qp} & \mathbf{0} \\[4pt]
\mathbf{J}_{pq} & \mathbf{J}_{pp} & \mathbf{0} \\[4pt]
\mathbf{0} & \mathbf{0} & \mathbf{0}
\end{pmatrix}
}_{\mathbf{J}}
\begin{pmatrix}
\nabla_{\mathbf{q}} H_{qp} \\[4pt] \nabla_{\mathbf{p}} H_{qp} \\[4pt] \nabla_{\mathbf{s}} H_s
\end{pmatrix}
-
\begin{pmatrix}
\mathbf{0} \\[4pt]
\nabla_{\mathbf{p}} \mathcal{F}_p(\mathbf{p}) \\[4pt]
\nabla_{\mathbf{s}} \mathcal{F}_s(\mathbf{s})
\end{pmatrix}
+
\begin{pmatrix}
\mathbf{0} \\[4pt]
\mathbf{J}_{pq}\nabla_{\mathbf{q}}H_s(\mathbf{q},\mathbf{s}) \\[4pt]
\mathbf{g}(\mathbf{x})\,\mathbf{u}
\end{pmatrix},
\label{eq:phs_dp}
\end{equation}
\begin{equation}
\mathbf{y} = \mathbf{g}^\top(\mathbf{x})\,\nabla_{\mathbf{s}} H_s(\mathbf{q},\mathbf{s}),
\label{eq:phs_dp_output}
\end{equation}
where $H_{qp} \in C^2(\mathbb{R}^{n_q}\times\mathbb{R}^{n_p};\mathbb{R})$ is the Hamiltonian of the $(\mathbf{q},\mathbf{p})$-subsystem, $H_s \in C^2(\mathbb{R}^{n_q}\times\mathbb{R}^{n_s};\mathbb{R})$ is the Hamiltonian of the $\mathbf{s}$-subsystem with $H_s$ affine in $\mathbf{q}$ so that $\nabla_{\mathbf{q}}H_s$ depends only on $\mathbf{s}$, $\mathbf{J} = -\mathbf{J}^\top \in \mathbb{R}^{n \times n}$ is the interconnection matrix with $\mathbf{J}_{pq} = -\mathbf{J}_{qp}^\top$, $\mathbf{g}(\mathbf{x}) \in \mathbb{R}^{n_s \times m}$ is the input map, $(\mathbf{u}, \mathbf{y}) \in \mathbb{R}^m \times \mathbb{R}^m$ are conjugate port variables whose product $\mathbf{y}^\top\mathbf{u}$ represents the power supplied to the $\mathbf{s}$-channel, and $\mathcal{F}_p \in C^2(\mathbb{R}^{n_p};\mathbb{R}_{\geq 0})$, $\mathcal{F}_s \in C^2(\mathbb{R}^{n_s};\mathbb{R}_{\geq 0})$ are the \emph{dissipation potentials} satisfying
\begin{equation}
\begin{aligned}
\nabla_{\mathbf{p}}\mathcal{F}_p(\mathbf{0}) = \mathbf{0}, \quad &\nabla^2_{\mathbf{p}}\mathcal{F}_p(\mathbf{p}) \succ \mathbf{0}, \\[4pt]
\nabla_{\mathbf{s}}\mathcal{F}_s(\mathbf{0}) = \mathbf{0}, \quad &\nabla^2_{\mathbf{s}}\mathcal{F}_s(\mathbf{s}) \succ \mathbf{0}.
\end{aligned}
\label{eq:F_convex}
\end{equation}
\end{definition}

\begin{remark}
The term $\mathbf{J}_{pq}\nabla_{\mathbf{q}}H_s(\mathbf{s})$ in~\eqref{eq:phs_dp} captures the action of the $\mathbf{s}$-subsystem on the $\mathbf{p}$-dynamics through the shared coordinate $\mathbf{q}$. Since $H_s$ is affine in $\mathbf{q}$, this term depends only on $\mathbf{s}$ and is denoted $\boldsymbol{\phi}(\mathbf{s}) \triangleq \mathbf{J}_{pq}\nabla_{\mathbf{q}}H_s(\mathbf{s})$ for subsequent use in the control design.
\end{remark}

The PHS-DP framework captures all physical systems whose interconnection structure satisfies the Jacobi identity, encompassing the full scope of the port-Hamiltonian paradigm. The decomposition $\mathbf{x} = (\mathbf{q}, \mathbf{p}, \mathbf{s})$ arises naturally whenever a mechanical subsystem is coupled to a non-mechanical one, with the coupling encoded through the dependence of $H_s$ on $\mathbf{q}$. Representative examples include magnetic levitation systems (flux linkage coupled to ball position), permanent magnet synchronous motors (stator flux linked to rotor angle), hydraulically actuated manipulators (chamber pressure coupled to joint configuration), and piezoelectric actuators (electric charge coupled to mechanical displacement). In each case, the zero rows in $\mathbf{J}$ corresponding to $\mathbf{s}$ reflect the absence of direct conservative interconnection between the non-mechanical and mechanical degrees of freedom, while $\mathbf{J}_{pq}\nabla_{\mathbf{q}}H_s(\mathbf{s})$ encodes the resulting generalised force on the mechanical subsystem. When $n_s = 0$, the framework reduces to a purely mechanical system, in which the input map $\mathbf{g}$ and output $\mathbf{y}$ act entirely through the $\mathbf{p}$-channel.

\begin{remark}
Condition~\eqref{eq:F_convex} is the scalar analogue of positive definiteness of the damping matrix in the standard PHS framework. When
\begin{equation}
\begin{aligned}
\mathcal{F}_p(\mathbf{p}) &= \frac{1}{2}(\nabla_{\mathbf{p}} H_{qp})^\top \mathbf{R}_p\, \nabla_{\mathbf{p}} H_{qp}, \\[4pt]
\mathcal{F}_s(\mathbf{s}) &= \frac{1}{2}(\nabla_{\mathbf{s}} H_s)^\top \mathbf{R}_s\, \nabla_{\mathbf{s}} H_s,
\end{aligned}
\label{eq:F_quadratic}
\end{equation}
with $\mathbf{R}_p = \mathbf{R}_p^\top \succ \mathbf{0}$ and $\mathbf{R}_s = \mathbf{R}_s^\top \succ \mathbf{0}$, system~\eqref{eq:phs_dp} reduces to the standard PHS~\eqref{eq:phs} with $\mathbf{R} = \mathrm{blkdiag}(\mathbf{0},\, \mathbf{R}_p,\, \mathbf{R}_s)$, confirming that PHS-DP strictly generalises the standard framework.
\end{remark}

\begin{proposition}
\label{prop:power_balance}
The PHS-DP~\eqref{eq:phs_dp}--\eqref{eq:phs_dp_output} satisfies
\begin{equation}
\dot{H} =
- \nabla_{\mathbf{p}} H_{qp}^\top \nabla_{\mathbf{p}} \mathcal{F}_p
- \nabla_{\mathbf{s}} H_s^\top \nabla_{\mathbf{s}} \mathcal{F}_s
+ \mathbf{y}^\top \mathbf{u}
\;\leq\; \mathbf{y}^\top \mathbf{u},
\label{eq:power_balance}
\end{equation}
and is therefore passive with storage function $H = H_{qp} + H_s$ and supply rate $\mathbf{y}^\top \mathbf{u}$.
\end{proposition}

\begin{proof}
Differentiating $H = H_{qp} + H_s$ along trajectories of~\eqref{eq:phs_dp} and using skew-symmetry of $\mathbf{J}$ gives
\begin{align*}
\dot{H}
&= \nabla_{\mathbf{q}} H_{qp}^\top \dot{\mathbf{q}}
+ \nabla_{\mathbf{p}} H_{qp}^\top \dot{\mathbf{p}}
+ \nabla_{\mathbf{q}} H_s^\top \dot{\mathbf{q}}
+ \nabla_{\mathbf{s}} H_s^\top \dot{\mathbf{s}} \\
&= \nabla H^\top \mathbf{J} \nabla H
- \nabla_{\mathbf{p}} H_{qp}^\top \nabla_{\mathbf{p}} \mathcal{F}_p
- \nabla_{\mathbf{s}} H_s^\top \nabla_{\mathbf{s}} \mathcal{F}_s \\
&\quad + \nabla_{\mathbf{p}} H_{qp}^\top \mathbf{J}_{pq}\nabla_{\mathbf{q}}H_s
+ \nabla_{\mathbf{q}} H_s^\top \mathbf{J}_{qp}\nabla_{\mathbf{p}}H_{qp}
+ \nabla_{\mathbf{s}} H_s^\top \mathbf{g}\,\mathbf{u} \\
&= - \nabla_{\mathbf{p}} H_{qp}^\top \nabla_{\mathbf{p}} \mathcal{F}_p
- \nabla_{\mathbf{s}} H_s^\top \nabla_{\mathbf{s}} \mathcal{F}_s
+ \mathbf{y}^\top \mathbf{u},
\end{align*}
where $\nabla H^\top \mathbf{J} \nabla H = 0$ by skew-symmetry of $\mathbf{J}$, the terms $\nabla_{\mathbf{p}} H_{qp}^\top \mathbf{J}_{pq}\nabla_{\mathbf{q}}H_s$ and $\nabla_{\mathbf{q}} H_s^\top \mathbf{J}_{qp}\nabla_{\mathbf{p}}H_{qp}$ cancel by skew-symmetry of $\mathbf{J}_{pq} = -\mathbf{J}_{qp}^\top$, and the last equality uses the output definition~\eqref{eq:phs_dp_output}. The inequality in~\eqref{eq:power_balance} follows from the convexity conditions~\eqref{eq:F_convex}, which imply $\nabla_{\mathbf{p}} H_{qp}^\top \nabla_{\mathbf{p}} \mathcal{F}_p \geq 0$ and $\nabla_{\mathbf{s}} H_s^\top \nabla_{\mathbf{s}} \mathcal{F}_s \geq 0$.
\end{proof}

\section{Dual Potential Shaping Control}
\label{sec:control}

This section develops the DPSC framework for trajectory tracking of PHS-DP systems~\eqref{eq:phs_dp} by sequentially shaping three scalar objects: the potential energy, the $\mathbf{p}$-channel dissipation potential, and the $\mathbf{s}$-channel dissipation potential, with the interconnection structure preserved throughout and no PDE to be solved. Throughout this section, $\mathbf{J}$ is assumed constant, the coupling term $\boldsymbol{\phi}(\mathbf{s})$ is assumed invertible on $\mathcal{D}$, and $\mathbf{x}^\star(t) = (\mathbf{q}^{\star\top}(t), \mathbf{p}^{\star\top}(t), \mathbf{s}^{\star\top}(t))^\top$ denotes a feasible trajectory of~\eqref{eq:phs_dp}, i.e.\ there exists $\mathbf{u}^\star(t)$ such that~\eqref{eq:phs_dp} holds at $\mathbf{x} = \mathbf{x}^\star(t)$ for all $t \geq 0$.

\paragraph{Step 1: Potential energy shaping.}
The potential energy component of $H_{qp}$ is replaced by a desired potential $V_d(\mathbf{q}, \mathbf{q}^\star(t))$, yielding the shaped $(\mathbf{q},\mathbf{p})$-subsystem Hamiltonian
\begin{equation}
H_{qp,d}(\mathbf{x}, \mathbf{q}^\star(t)) = T(\mathbf{p}) + V_d\!\left(\mathbf{q}, \mathbf{q}^\star(t)\right),
\label{eq:Hd}
\end{equation}
where $T(\mathbf{p})$ is the kinetic energy, unchanged by the design. The function $V_d$ is required to satisfy
\begin{equation}
\begin{aligned}
\nabla_{\mathbf{q}} V_d\!\left(\mathbf{q}^\star(t), \mathbf{q}^\star(t)\right) &= \mathbf{0}, \\
\nabla^2_{\mathbf{q}} V_d(\mathbf{q}, \mathbf{q}^\star(t)) &\succeq \lambda_q \mathbf{I} \succ \mathbf{0},
\quad \forall\, \mathbf{q},\, t \geq 0,
\end{aligned}
\label{eq:Vd_cond}
\end{equation}
so that $\mathbf{q}^\star(t)$ is the strict global minimum of $V_d(\cdot, \mathbf{q}^\star(t))$ at each $t$.

\paragraph{Step 2: Dissipation shaping on the $\mathbf{p}$-channel and determination of $\mathbf{s}_d(t)$.}
With $H_{qp,d}$ fixed, the desired dissipation potential on the $\mathbf{p}$-channel is
\begin{equation}
\mathcal{F}_{p,d}\!\left(\mathbf{p}, \mathbf{p}^\star(t)\right)
=
\mathcal{F}_{p,e}\!\left(\mathbf{p} - \mathbf{p}^\star(t)\right)
- \dot{\mathbf{p}}^{\star\top}(t)\,\mathbf{p},
\label{eq:Fp_star}
\end{equation}
where $\mathcal{F}_{p,e}$ is a strictly convex error feedback potential satisfying $\nabla_{\mathbf{p}}\mathcal{F}_{p,e}(\mathbf{0}) = \mathbf{0}$ and $\nabla^2_{\mathbf{p}}\mathcal{F}_{p,e} \succ \mathbf{0}$, and $-\dot{\mathbf{p}}^{\star\top}(t)\,\mathbf{p}$ is a feedforward term encoding the reference momentum rate. The desired $(\mathbf{q}, \mathbf{p})$-subsystem dynamics takes the form
\begin{equation}
\begin{pmatrix} \dot{\mathbf{q}} \\ \dot{\mathbf{p}} \end{pmatrix}
=
\begin{pmatrix} \mathbf{J}_{qq} & \mathbf{J}_{qp} \\ \mathbf{J}_{pq} & \mathbf{J}_{pp} \end{pmatrix}
\begin{pmatrix} \nabla_{\mathbf{q}} H_{qp,d} \\ \nabla_{\mathbf{p}} H_{qp,d} \end{pmatrix}
-
\begin{pmatrix} \mathbf{0} \\ \nabla_{\mathbf{p}} \mathcal{F}_{p,d}(\mathbf{p}, \mathbf{p}^\star(t)) \end{pmatrix},
\label{eq:qp_cl}
\end{equation}
which is independent of $\mathbf{s}$ and constitutes the upper subsystem of the closed-loop cascade. The desired $\mathbf{s}$-trajectory required to realise~\eqref{eq:qp_cl} is obtained by inverting $\boldsymbol{\phi}$ through the $\mathbf{p}$-row:
\begin{align}
\mathbf{s}_d(t) = \boldsymbol{\phi}^{-1}\!\Bigl(
&\nabla_{\mathbf{p}}\mathcal{F}_p(\mathbf{p})
- \nabla_{\mathbf{p}} \mathcal{F}_{p,d}
+ \mathbf{J}_{pq}\,\nabla_{\mathbf{q}} H_{qp,d} \notag\\
&- \mathbf{J}_{pq}\,\nabla_{\mathbf{q}} H_{qp}
- \mathbf{J}_{pp}\,\nabla_{\mathbf{p}} H_{qp}
\Bigr).
\label{eq:s_d}
\end{align}

\paragraph{Step 3: Dissipation shaping on the $\mathbf{s}$-channel and control law.}
The desired dissipation potential on the $\mathbf{s}$-channel is
\begin{equation}
\mathcal{F}_{s,d}\!\left(\mathbf{s}, \mathbf{s}_d(t)\right)
=
\mathcal{F}_{s,e}\!\left(\mathbf{s} - \mathbf{s}_d(t)\right)
- \dot{\mathbf{s}}^{\star\top}(t)\,\mathbf{s},
\label{eq:Fs_star}
\end{equation}
where $\mathcal{F}_{s,e}$ is a strictly convex error feedback potential satisfying $\nabla_{\mathbf{s}}\mathcal{F}_{s,e}(\mathbf{0}) = \mathbf{0}$ and $\nabla^2_{\mathbf{s}}\mathcal{F}_{s,e} \succ \mathbf{0}$, and $-\dot{\mathbf{s}}^{\star\top}(t)\,\mathbf{s}$ is a feedforward term encoding the reference $\mathbf{s}$-trajectory rate. The closed-loop system takes the cascade form
\begin{align}
\begin{pmatrix} \dot{\mathbf{q}} \\ \dot{\mathbf{p}} \end{pmatrix}
&=
\begin{pmatrix} \mathbf{J}_{qq} & \mathbf{J}_{qp} \\ \mathbf{J}_{pq} & \mathbf{J}_{pp} \end{pmatrix}
\begin{pmatrix} \nabla_{\mathbf{q}} H_{qp,d} \\ \nabla_{\mathbf{p}} H_{qp,d} \end{pmatrix}
-
\begin{pmatrix} \mathbf{0} \\ \nabla_{\mathbf{p}} \mathcal{F}_{p,d}(\mathbf{p}, \mathbf{p}^\star(t)) \end{pmatrix},
\label{eq:cascade_qp} \\[6pt]
\dot{\mathbf{s}}
&= - \nabla_{\mathbf{s}} \mathcal{F}_{s,d}\!\left(\mathbf{s}, \mathbf{s}_d(t)\right),
\label{eq:cascade_s}
\end{align}
where~\eqref{eq:cascade_qp} is independent of $\mathbf{s}$. Since $\mathbf{u}$ acts only on the $\mathbf{s}$-row, equating~\eqref{eq:phs_dp} with~\eqref{eq:cascade_qp}--\eqref{eq:cascade_s} yields the control law
\begin{equation}
\mathbf{u} =
\bigl(\mathbf{g}^\top\mathbf{g}\bigr)^{-1}\mathbf{g}^\top\Bigl(
\nabla_{\mathbf{s}} \mathcal{F}_s(\mathbf{s})
- \nabla_{\mathbf{s}} \mathcal{F}_{s,d}\!\left(\mathbf{s}, \mathbf{s}_d(t)\right)
\Bigr),
\label{eq:control}
\end{equation}
which requires only the gradients of the open-loop and shaped dissipation potentials on the $\mathbf{s}$-channel.

\begin{theorem}[DPSC Tracking]
\label{thm:tracking}
Consider the PHS-DP~\eqref{eq:phs_dp} with constant $\mathbf{J}$, $\mathbf{g}$, inertia matrix $\boldsymbol{\mathcal{M}}(\mathbf{q}) \succ \mathbf{0}$, and feasible trajectory $\mathbf{x}^\star(t)$. Suppose $V_d$, $\mathcal{F}_{p,e}$, and $\mathcal{F}_{s,e}$ satisfy
\begin{equation}
\begin{aligned}
\nabla_{\mathbf{q}} V_d\!\left(\mathbf{q}^\star(t), \mathbf{q}^\star(t)\right) &= \mathbf{0}, \quad
\nabla^2_{\mathbf{q}} V_d \succeq \lambda_q \mathbf{I}, \\
\nabla_{\mathbf{p}} \mathcal{F}_{p,e}(\mathbf{0}) = \mathbf{0}, \quad
&\nabla^2_{\mathbf{p}} \mathcal{F}_{p,e} \succeq \lambda_p \mathbf{I}, \\
\nabla_{\mathbf{s}} \mathcal{F}_{s,e}(\mathbf{0}) = \mathbf{0}, \quad
&\nabla^2_{\mathbf{s}} \mathcal{F}_{s,e} \succeq \lambda_s \mathbf{I},
\end{aligned}
\label{eq:convexity}
\end{equation}
for some $\lambda_q, \lambda_p, \lambda_s > 0$ and all $\mathbf{x} \in \mathcal{D}$, $t \geq 0$. Suppose further that the feasibility condition
\begin{equation}
\begin{pmatrix} \dot{\mathbf{q}}^\star \\ \dot{\mathbf{p}}^\star \\ \dot{\mathbf{s}}^\star \end{pmatrix}
=
\begin{pmatrix} \mathbf{J}_{qq} & \mathbf{J}_{qp} & \mathbf{0} \\ \mathbf{J}_{pq} & \mathbf{J}_{pp} & \mathbf{0} \\ \mathbf{0} & \mathbf{0} & \mathbf{0} \end{pmatrix}
\begin{pmatrix} \nabla_{\mathbf{q}} H_{qp,d}(\mathbf{x}^\star) \\ \nabla_{\mathbf{p}} H_{qp,d}(\mathbf{x}^\star) \\ \nabla_{\mathbf{s}} H_s(\mathbf{x}^\star) \end{pmatrix}
-
\begin{pmatrix} \mathbf{0} \\ \nabla_{\mathbf{p}} \mathcal{F}_{p,d}(\mathbf{p}^\star, \mathbf{p}^\star(t)) \\ \nabla_{\mathbf{s}} \mathcal{F}_{s,d}(\mathbf{s}^\star(t), \mathbf{s}_d(t)) \end{pmatrix}
\label{eq:feasibility}
\end{equation}
holds for all $t \geq 0$. Then the control law~\eqref{eq:control} renders $\mathbf{x}^\star(t)$ exponentially stable: all trajectories initialised sufficiently close to $\mathbf{x}^\star(0)$ converge exponentially to $\mathbf{x}^\star(t)$.
\end{theorem}

\begin{proof}
Substituting~\eqref{eq:control} into~\eqref{eq:phs_dp} yields the closed-loop cascade~\eqref{eq:cascade_qp}--\eqref{eq:cascade_s}, with $\mathbf{x}^\star(t)$ a trajectory thereof by~\eqref{eq:feasibility}. The virtual dynamics of the cascade takes the lower-triangular form
\begin{equation}
\frac{d}{dt}
\begin{pmatrix} \delta\mathbf{x}_{qp} \\ \delta\mathbf{s} \end{pmatrix}
=
\underbrace{\begin{pmatrix}
\mathbf{F}_{11} & \mathbf{0} \\
\mathbf{F}_{21} & \mathbf{F}_{22}
\end{pmatrix}}_{\mathbf{F}}
\begin{pmatrix} \delta\mathbf{x}_{qp} \\ \delta\mathbf{s} \end{pmatrix},
\label{eq:virtual_full}
\end{equation}
where $\delta\mathbf{x}_{qp} = (\delta\mathbf{q}^\top, \delta\mathbf{p}^\top)^\top$ and the Jacobian blocks are
\begin{align}
\mathbf{F}_{11} &= \begin{pmatrix} \mathbf{J}_{qq} & \mathbf{J}_{qp} \\ \mathbf{J}_{pq} & \mathbf{J}_{pp} \end{pmatrix}
\begin{pmatrix} \nabla^2_{\mathbf{q}}V_d & \mathbf{0} \\
\mathbf{0} & \boldsymbol{\mathcal{M}}^{-1} \end{pmatrix}
-
\begin{pmatrix} \mathbf{0} & \mathbf{0} \\
\mathbf{0} & \nabla^2_{\mathbf{p}}\mathcal{F}_{p,e} \end{pmatrix},
\label{eq:F11} \\[4pt]
\mathbf{F}_{22} &= -\nabla^2_{\mathbf{s}}\mathcal{F}_{s,e} \preceq -\lambda_s\mathbf{I},
\label{eq:F22}\\[4pt]
\mathbf{F}_{21} &= -\nabla^2_{\mathbf{s}}\mathcal{F}_{s,e}(\mathbf{s}-\mathbf{s}_d(t))\,\frac{\partial \mathbf{s}_d}{\partial \mathbf{x}_{qp}},
\label{eq:F21}
\end{align}
where $\mathbf{J}_{qq} = \mathbf{0}$ for port-Hamiltonian systems with standard symplectic structure, and by differentiating~\eqref{eq:s_d},
\begin{equation}
\frac{\partial \mathbf{s}_d}{\partial \mathbf{x}_{qp}} = \left(\frac{\partial \boldsymbol{\phi}}{\partial \mathbf{s}}\right)^{-1}
\begin{pmatrix}
\mathbf{J}_{pq}\bigl(\nabla^2_{\mathbf{q}}H_{qp,d} - \nabla^2_{\mathbf{q}}H_{qp}\bigr) \\
\nabla^2_{\mathbf{p}}\mathcal{F}_p - \nabla^2_{\mathbf{p}}\mathcal{F}_{p,d} - \mathbf{J}_{pp}\nabla^2_{\mathbf{p}}H_{qp}
\end{pmatrix}^\top.
\label{eq:dsd_dxqp}
\end{equation}

\emph{$(\mathbf{q},\mathbf{p})$-subsystem.}
For any eigenvalue $\lambda$ of $\mathbf{F}_{11}$ with eigenvector $\mathbf{v} = (\mathbf{v}_q^\top, \mathbf{v}_p^\top)^\top \neq \mathbf{0}$, premultiplying $\mathbf{F}_{11}\mathbf{v} = \lambda\mathbf{v}$ by $\mathbf{v}^*\boldsymbol{\Phi}^{-1}$ with $\boldsymbol{\Phi} \triangleq \mathrm{blkdiag}(\nabla^2_{\mathbf{q}}V_d,\,\boldsymbol{\mathcal{M}}^{-1})$ and taking the real part gives
\begin{equation}
\mathrm{Re}(\lambda)\,\mathbf{v}^*\boldsymbol{\Phi}^{-1}\mathbf{v}
= -\mathbf{v}_p^*\nabla^2_{\mathbf{p}}\mathcal{F}_{p,e}\,\boldsymbol{\mathcal{M}}\,\mathbf{v}_p \leq 0,
\label{eq:re_eig}
\end{equation}
where the skew-symmetric contribution of $\mathbf{J}_{qp}^{\mathrm{blk}}$ vanishes. If $\mathrm{Re}(\lambda) = 0$, then $\mathbf{v}_p = \mathbf{0}$, and the $\mathbf{q}$-row of $\mathbf{F}_{11}\mathbf{v} = \lambda\mathbf{v}$ with $\mathbf{J}_{qq} = \mathbf{0}$ reduces to $\lambda\mathbf{v}_q = \mathbf{J}_{qp}\boldsymbol{\mathcal{M}}^{-1}\mathbf{v}_p = \mathbf{0}$; since $\mathbf{v} \neq \mathbf{0}$ this forces $\lambda = 0$, contradicting the uniform convexity of $V_d$ which precludes a zero eigenvalue. Hence $\mathbf{F}_{11}$ is Hurwitz, and the unique positive definite solution $\mathbf{M}_{qp} \succ \mathbf{0}$ of the Lyapunov equation $\mathbf{M}_{qp}\mathbf{F}_{11} + \mathbf{F}_{11}^\top\mathbf{M}_{qp} = -\mathbf{I}$ satisfies
\begin{equation}
\mathbf{M}_{qp}\mathbf{F}_{11} + \mathbf{F}_{11}^\top\mathbf{M}_{qp} = -\mathbf{I} \preceq -\frac{1}{\lambda_{\max}(\mathbf{M}_{qp})}\mathbf{M}_{qp} = -2\beta_{qp}\mathbf{M}_{qp},
\label{eq:contraction_qp}
\end{equation}
where $\beta_{qp} \triangleq 1/(2\lambda_{\max}(\mathbf{M}_{qp})) > 0$. Hence~\eqref{eq:cascade_qp} is contracting at rate $\beta_{qp}$ in metric $\mathbf{M}_{qp}$.

\emph{$\mathbf{s}$-subsystem.}
From~\eqref{eq:F22}, \eqref{eq:cascade_s} is contracting at rate $\lambda_s$ in the identity metric, uniformly in $\mathbf{s}_d(t)$. The coupling block $\mathbf{F}_{21}$ is bounded on $\mathcal{D}$: $\nabla^2_{\mathbf{s}}\mathcal{F}_{s,e}$ is bounded by $C^2$ smoothness, and $\partial\mathbf{s}_d/\partial\mathbf{x}_{qp}$ is bounded by~\eqref{eq:dsd_dxqp} together with the $C^2$ regularity of $H_{qp}$, $H_{qp,d}$, $\mathcal{F}_p$, $\mathcal{F}_{p,d}$ and the invertibility of $\partial\boldsymbol{\phi}/\partial\mathbf{s}$ on $\mathcal{D}$.

\emph{Conclusion.}
By Proposition~\ref{prop:hierarchical}, the lower-triangular virtual dynamics~\eqref{eq:virtual_full} with contracting diagonal blocks and bounded coupling is contracting, so all trajectories initialised sufficiently close to $\mathbf{x}^\star(0)$ converge exponentially to $\mathbf{x}^\star(t)$.
\end{proof}

\section{Application to Magnetic Levitation}
\label{sec:sim}

The magnetic levitation system is modelled in PHS-DP form~\eqref{eq:phs_dp} with state $(q, p, s)^\top$, where $q$ is the ball height, $p$ the vertical momentum, and $s$ the magnetic flux linkage. The two subsystem Hamiltonians are
\begin{equation}
H_{qp}(q, p) = \frac{p^2}{2m} + bq, \qquad
H_s(q, s) = \frac{(c - q)}{2k}\,s^2,
\label{eq:H_maglev}
\end{equation}
and the open-loop dissipation potential on the $s$-channel is
\begin{equation}
\mathcal{F}_s(s) = \frac{r_2(c - q)}{2k}\,s^2.
\label{eq:F_maglev}
\end{equation}
The system takes the PHS-DP form~\eqref{eq:phs_dp} with structural matrices
\begin{equation}
\mathbf{J} =
\begin{pmatrix} 0 & 1 & 0 \\ -1 & 0 & 0 \\ 0 & 0 & 0 \end{pmatrix},
\qquad
g = 1,
\qquad
\boldsymbol{\phi}(s) = \frac{s^2}{2k},
\label{eq:sys_maglev}
\end{equation}
where $g = 1$ is the scalar input gain acting on the flux channel $\dot{s}$, $\boldsymbol{\phi}(s)$ is the electromagnetic force acting on the ball, $r_2 > 0$ is the coil resistance, and $b, c, k, m > 0$ are physical constants. The convexity condition~\eqref{eq:F_convex} requires
\begin{equation}
\frac{\partial^2 \mathcal{F}_s}{\partial s^2} = \frac{r_2(c - q)}{k} > 0,
\end{equation}
which holds under the physical constraint $q < c$. The conjugate output~\eqref{eq:phs_dp_output} is
\begin{equation}
y = \nabla_{s} H_s = \frac{(c-q)}{k}\,s,
\label{eq:output_maglev}
\end{equation}
corresponding to the coil current, so that $yu$ is the electrical power delivered to the coil. The control objective is to track a reference height $f(t)$ satisfying $m\ddot{f}(t) > -b$ for all $t \geq 0$, which ensures feasibility. The reference trajectory is
\begin{equation}
\begin{aligned}
q^\star(t) &= f(t), \qquad
p^\star(t) = m\dot{f}(t), \\
s^\star(t) &= \sqrt{2k\bigl(m\ddot{f}(t) + b\bigr)}.
\end{aligned}
\label{eq:ref_maglev}
\end{equation}

\paragraph{Step 1: Potential energy shaping.}
The potential energy $bq$ is replaced by the desired potential
\begin{equation}
V_d\!\left(q, f(t)\right) = \frac{k_q}{2}\bigl(q - f(t)\bigr)^2,
\qquad k_q > 0,
\label{eq:Vd_maglev}
\end{equation}
yielding the shaped $(\mathbf{q},\mathbf{p})$-subsystem Hamiltonian
\begin{equation}
H_{qp,d} = \frac{p^2}{2m} + \frac{k_q}{2}\bigl(q - f(t)\bigr)^2.
\label{eq:Hd_maglev}
\end{equation}
Condition~\eqref{eq:Vd_cond} is satisfied since $\nabla_{q}V_d|_{q = f(t)} = 0$ and $\partial^2 V_d/\partial q^2 = k_q > 0$.

\paragraph{Step 2: Dissipation shaping on the $p$-channel and determination of $s_d(t)$.}
Since the $p$-channel carries no open-loop dissipation, the desired dissipation potential is
\begin{equation}
\mathcal{F}_{p,d}\!\left(p, p^\star(t)\right)
= \frac{k_p}{2}\bigl(p - p^\star(t)\bigr)^2
- \dot{p}^{\star}(t)\,p,
\qquad k_p > 0.
\label{eq:Fp_maglev}
\end{equation}
Applying~\eqref{eq:s_d} and inverting $\boldsymbol{\phi}$ gives the desired flux trajectory
\begin{equation}
s_d(t) = \sqrt{2k\Bigl[
m\ddot{f}(t) + b
- k_p\bigl(p - p^\star(t)\bigr)
- k_q\bigl(q - f(t)\bigr)
\Bigr]},
\label{eq:sd_maglev}
\end{equation}
which encodes both the position error $q - f(t)$ and the momentum error $p - p^\star(t)$ through the shaped potentials.

\paragraph{Step 3: Dissipation shaping on the $s$-channel and control law.}
The desired dissipation potential on the $s$-channel is
\begin{equation}
\mathcal{F}_{s,d}\!\left(s, s_d(t)\right)
= \frac{k_s}{2}\bigl(s - s_d(t)\bigr)^2
- \dot{s}^{\star}(t)\,s,
\qquad k_s > 0,
\label{eq:Fs_maglev}
\end{equation}
where the feedforward rate is
\begin{equation}
\dot{s}^{\star}(t) =
\frac{km\dddot{f}(t)}{\sqrt{2k\bigl(m\ddot{f}(t) + b\bigr)}}.
\label{eq:sdot_maglev}
\end{equation}
The control law~\eqref{eq:control} yields
\begin{equation}
u = \nabla_{s}\mathcal{F}_s(s)
- \nabla_{s}\mathcal{F}_{s,d}\!\left(s, s_d(t)\right),
\label{eq:u_maglev_compact}
\end{equation}
which expands to
\begin{equation}
u = \frac{r_2(c - q)}{k}\,s
- k_s\bigl(s - s_d(t)\bigr)
+ \dot{s}^{\star}(t).
\label{eq:u_maglev}
\end{equation}

\begin{figure}[t]
    \centering
    \includegraphics[width=\columnwidth]{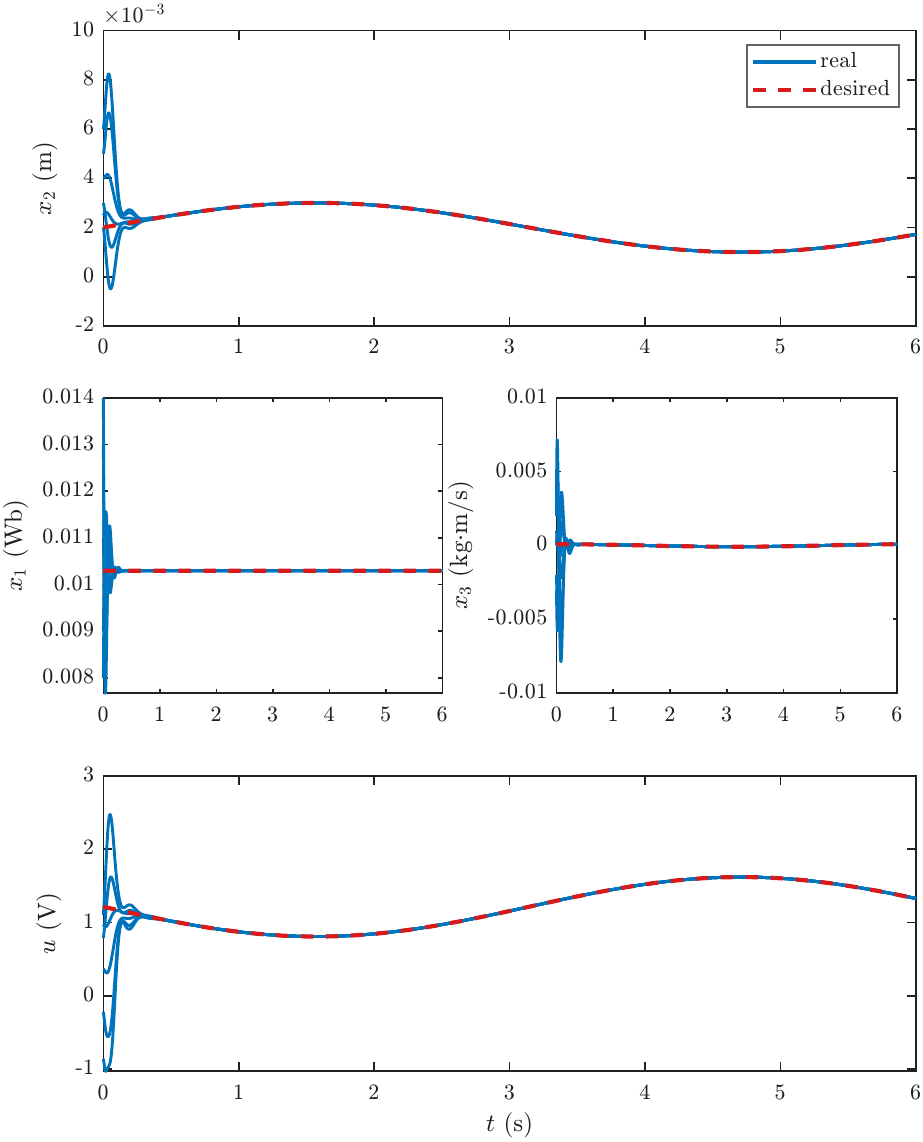}
    \caption{Trajectory tracking under multiple initial conditions via DPSC.}
    \label{fig:sim_maglev}
\end{figure}

\begin{figure}[t]
\centering
\includegraphics[width=\columnwidth]{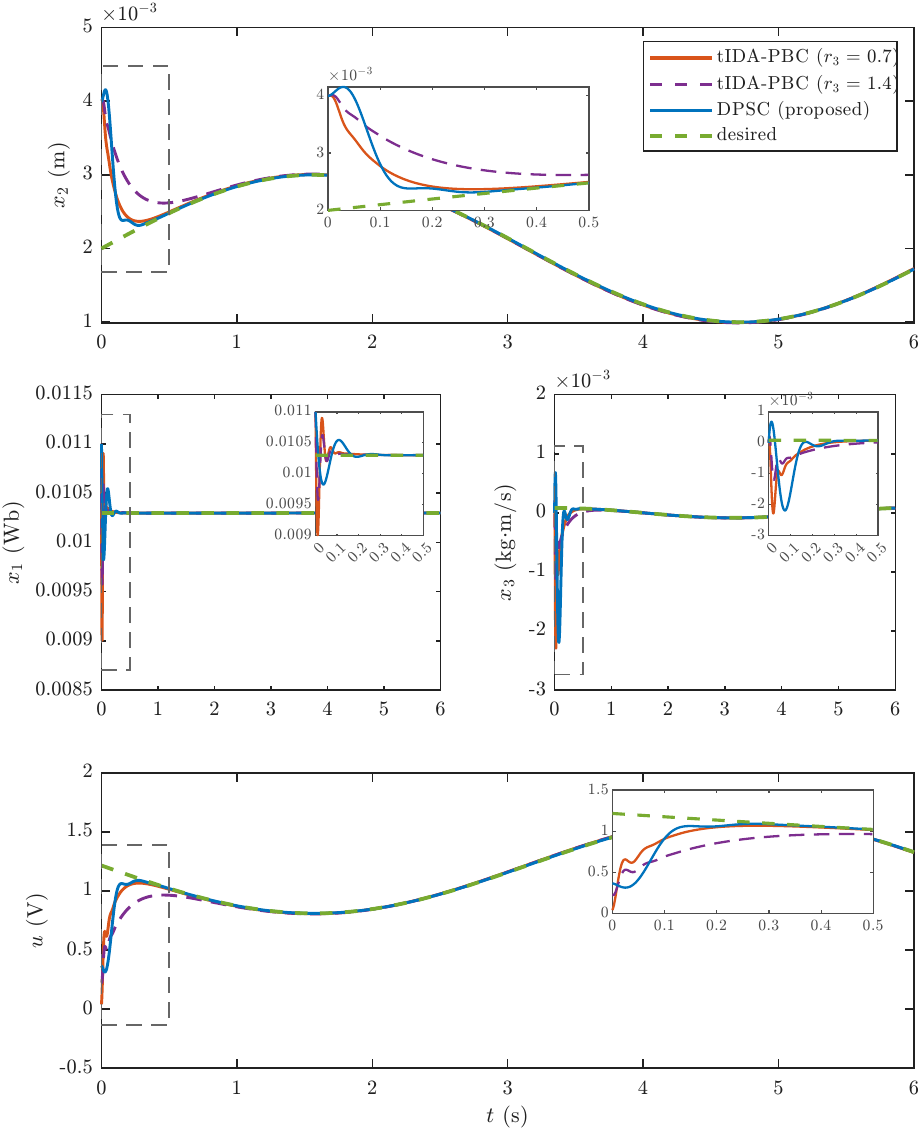}
\caption{Comparison of state trajectories and control input between tIDA-PBC ($r_3=0.7$), tIDA-PBC ($r_3=1.4$), and the proposed DPSC method.}
\label{fig:comparison_tracking}
\end{figure}

The first term cancels the open-loop resistive dissipation, the second injects error feedback driving $s \to s_d(t)$, and the third provides the reference feedforward.

Each shaped potential in the DPSC design carries an explicit physical interpretation: $V_d$ encodes the desired mechanical equilibrium, $\mathcal{F}_{p,d}$ shapes the momentum dynamics while preserving the kinetic energy structure, and $\mathcal{F}_{s,d}$ drives the electromagnetic subsystem to its desired trajectory. Throughout the design, the interconnection structure $\mathbf{J}$ is preserved entirely, respecting the physical coupling between the mechanical and electromagnetic subsystems inherent to the system.

In the DPSC framework, the matching condition is satisfied automatically. Since $u$ acts only on the $s$-row of~\eqref{eq:phs_dp}, the $(q,p)$-subsystem dynamics is unaffected by the control input, and the open-loop and closed-loop $(q,p)$-dynamics coincide on the unactuated channels:
\begin{equation}
\begin{pmatrix} \dot{q} \\ \dot{p} \end{pmatrix}
=
\begin{pmatrix} 0 & 1 \\ -1 & 0 \end{pmatrix}
\begin{pmatrix} \nabla_{q} H_{qp} \\ \nabla_{p} H_{qp} \end{pmatrix}
=
\begin{pmatrix} 0 & 1 \\ -1 & 0 \end{pmatrix}
\begin{pmatrix} \nabla_{q} H_{qp,d} \\ \nabla_{p} H_{qp,d} \end{pmatrix}
-
\begin{pmatrix} 0 \\ \nabla_{p}\mathcal{F}_{p,d} \end{pmatrix},
\label{eq:matching_auto}
\end{equation}
which is satisfied by construction through the choice of $s_d(t)$~\eqref{eq:sd_maglev}, without solving any PDE. This eliminates the principal obstacle of IDA-PBC-based methods and makes the design entirely algebraic.

The system parameters are taken from~\cite{Yaghmaei2017}:
\begin{equation*}
\begin{aligned}
&k = 6.4042\times10^{-5}\ \mathrm{N\,m/A^2},
\quad
r_2 = 2.52\ \Omega,
\quad
c = 0.005\ \mathrm{m},
\\
&b = 0.8280\ \mathrm{kg\,m/s^2},
\quad
m = 0.0844\ \mathrm{kg},
\end{aligned}
\end{equation*}
with controller gains $k_q = 50$, $k_p = 50$, $k_s = 50$. The reference trajectory is $f(t) = 10^{-3}(2 + \sin t)\ \mathrm{m}$, consistent with~\cite{Yaghmaei2017}. Simulations are initialised from multiple initial conditions to demonstrate the contraction property. The results are shown in Fig.~\ref{fig:sim_maglev}.

Fig.~\ref{fig:comparison_tracking} compares the closed-loop trajectories under the initial condition $(q_0, p_0, s_0)^\top = (0.011,\ 0.004,\ 0)^\top$. When $r_3 = 1.4$, tIDA-PBC converges slowly in $q$ and undershoots the reference during the transient. Reducing to $r_3 = 0.7$ improves the convergence rate to a level comparable to the proposed method, but introduces a pronounced overshoot in the flux $s$ and a larger initial excursion in the control input $u$. The proposed DPSC achieves similar convergence speed in $q$ with a smoother flux transient, remaining close to the desired trajectory throughout. All three controllers reach steady-state tracking by $t \approx 1\ \mathrm{s}$.

\section{Conclusion}
\label{sec:con}

This paper has proposed the PHS-DP framework and the DPSC method for trajectory tracking of port-Hamiltonian systems. By replacing the damping matrix with a scalar convex dissipation potential, PHS-DP places stored and dissipated energy on equal footing as independent physical objects, endowing the cotangent space with a dedicated scalar potential consistent with the variational structure of Hamiltonian mechanics. Building on this foundation, DPSC achieves trajectory tracking by sequentially shaping the potential energy and dissipation potentials in accordance with the desired trajectory, with exponential convergence established via a hierarchical contraction argument. The matching condition is satisfied automatically for any admissible choice of shaped potentials, eliminating the need to solve a PDE, and every step of the design retains a transparent physical interpretation. Application to a magnetic levitation system demonstrates tracking performance comparable to tIDA-PBC with substantially reduced design complexity. Future work will extend the PHS-DP framework to systems evolving on Lie groups, broadening its applicability to platforms such as UAVs and AUVs whose configuration spaces are inherently non-Euclidean.

\section*{Credit authorship contribution statement} 
$\bf{Jinjun\ Jia}$: Conceptualization, Methodology, Writing - original draft, Software. $\bf{Yuchen\ Liao}$: Data curation, Investigation. $\bf{Kang\ An}$: Investigation, Software. $\bf{Xun\ Yan}$: Visualization, Validation. $\bf{Tiedong\ Zhang}$: Supervision, Funding acquisition. $\bf{Dapeng\ Jiang}$: Conceptualization, Methodology, Funding acquisition.
\section*{Declaration of competing interest}
The authors declare that they have no known competing financial 
interests or personal relationships that could have appeared to influence the work reported in this paper.

\section*{Acknowledgements}                              
Project supported by Southern Marine Science and Engineering Guangdong Laboratory (Zhuhai) (Grant No. SML2024SP007), Department of Science and Technology of Guangdong Province (Grant No. 2025B1111130002), National Key Research and Development Program of China (Grant No. 2024YFB4710800, 2024YFB4710803, 2024YFB4710805), National Natural Science Foundation of China (Grant No. 52571375, U22A2012), the Development Programme Project of Heilongjiang Province (Grant No. GA20A402).  

\bibliographystyle{elsarticle-harv}
\bibliography{ref.bib}

\end{document}